\documentclass[fleqn,submission,copyright,creativecommons]{eptcs} 
\providecommand{\event}{EXPRESS/SOS 2023}

\usepackage{iftex}

\ifpdf
  \usepackage{underscore}         
  \usepackage[T1]{fontenc}        
\else
  \usepackage{breakurl}           
\fi

\usepackage{hyperref}
\usepackage[hyphenbreaks]{breakurl}
\usepackage{graphicx}
\usepackage{color}
\usepackage{amsmath,amssymb,amsthm}
\usepackage{calc,xspace,stmaryrd,times}
\usepackage{procalg}
\usepackage{tikz}
\usetikzlibrary{arrows,shapes,automata,positioning}
\tikzset{
        state/.style={
          circle,
          draw,
          minimum size=6mm,
    },
}
\usepackage{ulem}
\newtheorem{defi}{Definition}

\newtheorem{thm}{Theorem}

\newtheorem{exa}{Example}

\newcommand{\Procids}[1][\mathcal{P}]{\ensuremath{#1}}

\newcommand{\pref}[2][a]{\ensuremath{{#1}.{#2}}}

\renewcommand{\merge}{\mathrel{\parallel}}
\newcommand{\PA}{\ensuremath{\mathrm{PA}\:}}
\newcommand{\BPP}{\ensuremath{\mathrm{BPP}\:}}
\newcommand{\BCP}{\ensuremath{\mathrm{BCP}\:}}

\newcommand{\nstep}[1]{\ensuremath{\overset{#1}{\nrightarrow}}}
\newcommand{\nbisim}[1][]{%
    \setbox0=\hbox{\kern-.1ex{$\leftrightarrow$}\kern-.1ex}
    \setbox1=\vbox{\hbox{\raise .1ex \box0}\hrule}%
    \ensuremath{\not\mathrel{\hbox{\kern.1ex\box1\kern.1ex}_{#1}}}
  }

\newcommand{\defeqn}{\ensuremath{\mathrel{\stackrel{\textrm{def}}{=}}}}

\title{Parallel Pushdown Automata and Commutative Context-Free
  Grammars in Bisimulation Semantics (Extended Abstract)}

\author{Jos C. M. Baeten
\institute{CWI\\ Amsterdam, The Netherlands}
\email{Jos.Baeten@cwi.nl} 
\and
  Bas Luttik
\institute{Eindhoven University of Technology\\ Eindhoven, The
  Netherlands}
\email{s.p.luttik@tue.nl}
}
\def\titlerunning{Parallel PDA and Commutative CFG in Bisimulation Semantics}
\def\authorrunning{J. C. M. Baeten \& B. Luttik}

\begin{document}
\maketitle

\begin{abstract}
A classical theorem states that the set of languages given by a
pushdown automaton coincides with the set of languages given by a
context-free grammar. In previous work, we proved the pendant of this theorem in a setting with interaction: the set of processes given by a pushdown automaton coincides with the set of processes given by a finite guarded recursive specification over a process algebra with actions, choice, sequencing and guarded recursion, if and only if we add sequential value passing. In this paper, we look what happens if we consider parallel pushdown automata instead of pushdown automata, and a process algebra with parallelism instead of sequencing.
\end{abstract}

\section{Introduction}

This paper contributes to our ongoing project to integrate the theory
of automata and formal languages on the one hand and concurrency
theory on the other hand. The integration requires a more refined view
on the semantics of automata, grammars and expressions. Instead 
of treating automata as language acceptors, and grammars and
expressions as syntactic means to specify languages, we propose to view
them both as defining process graphs. The great benefit of this
approach is that process graphs can be considered modulo a plethora of
behavioural equivalences \cite{Gla93}. One can still consider language
equivalence and recover the classical theory of automata and formal
languages. But one can also consider finer notions such as
bisimilarity, which is better suited for interacting processes.

The project started with a structural characterisation of the class of
finite automata of which the processes are denoted by regular
expressions up to bisimilarity \cite{BCG07}. The investigation of the
expressiveness of regular expressions in bisimulation semantics was
continued in \cite{BLMvT16}.
In \cite{BLT13}, we replaced the Turing machine as an abstract model
of a computer by the Reactive Turing Machine, which has interaction as
an essential ingredient. Transitions have labels to give a notion of
interactivity, and we consider the resulting process graphs modulo bisimilarity
rather than language equivalence. Thus a Reactive Turing Machine defines an
\emph{executable} interactive process, refining the notion of computable
function.

In the same way as classical automata theory defines a hierarchy of
formal languages, we obtain a hierarchy of processes. In \cite{BCL22},
we proved that the set of processes given by a pushdown automaton
coincides with the set of processes given by a finite guarded
recursive specification over a process algebra with actions, choice,
sequencing and guarded recursion, if and only if we add sequential
value passing. Pushdown automata provide an abstract model of a
computer with a memory in the form of a stack. In this paper, we
consider the abstract model of a computer with a memory in the form of
a bag.  We consider the correspondence between parallel pushdown
automata and commutative context-free grammars. In the process
setting, a commutative context-free grammar is a process algebra
comprising actions, choice, parallelism and recursion. We start out
from the process algebra \BPP, extended with constants for acceptance
and non-acceptance (deadlock).

Then we find that in one direction, every process of a finite guarded recursive specification over this process algebra is the process of a parallel pushdown automaton, but not the other way around: there are parallel pushdown automata with a process that is not the process of any finite guarded recursive specification. This is even the case for the one-state parallel pushdown automaton of the bag itself, there is no finite guarded \BPP-specification for it. If we do want to get a recursive specification for the bag, we need to give some actions priority over others, and can find a satisfactory specification over \BPP$_{\theta}$, \BPP extended with the priority operator. Indeed, we can obtain a finite guarded specification over \BPP$_{\theta}$ for every one-state parallel pushdown automaton. On the other hand, there is a parallel pushdown automaton with two states that does not have a finite guarded specification over \BPP$_{\theta}$.

 If we add communication with value passing to this algebra, resulting in \BCP$_{\theta}$, we do get a complete correspondence: a process is the process of a parallel pushdown automaton if and only if it is the process of a finite guarded recursive specification. We can also get this result in a setting without the priority operator, so over \BCP, but then we need
  that the set of values can be countable, and we have also countable summation.

To conclude, we provide a characterisation of parallel pushdown processes as a regular process communicating with a bag. In the case without priority operator, we need a form of asymmetric communication.

\section{Preliminaries}

As a common semantic framework we use the notion of a
\emph{labelled transition system}.

\begin{defi} \label{def:tsspace}
A \emph{labelled transition system} is a quadruple
$(\mathcal{S},\mathcal{A},{\xrightarrow{}},{\downarrow})$, where
\begin{enumerate}
\item $\mathcal{S}$ is a set of \emph{states};
\item $\mathcal{A}$ is a set of \emph{actions}, $\tau\not\in\mathcal{A}$ is the \emph{unobservable} or \emph{silent} action;
\item
${\xrightarrow{}}\subseteq{\mathcal{S}\times\mathcal{A}\cup\{\tau\}\times\mathcal{S}}$ is
an $\mathcal{A}\cup\{\tau\}$-labelled \emph{transition relation}; and
     \item ${\downarrow}\subseteq\mathcal{S}$ is the set of \emph{final} or
       \emph{accepting} states.
 \end{enumerate}
A \emph{process graph} is a
labelled transition system with a special
designated \emph{root state} ${\uparrow}$, i.e., it is a quintuple
$(\mathcal{S},\mathcal{A},{\rightarrow},{\uparrow},{\downarrow})$ such that
$(\mathcal{S},\mathcal{A},{\rightarrow},{\downarrow})$ is a labelled transition system, and ${\uparrow}\in\mathcal{S}$.
We write $s\xrightarrow{a}s'$ for $(s,a,s')\in{\rightarrow}$
and $\term{s}$ for $s\in\mathalpha{\downarrow}$.
\end{defi}

For $w\in\Act^{*}$ we define
$s\steps{w}t$ inductively, for all states $s,t,u$: first, $s \steps{\epsilon} s$, and then, for $a \in \Act$, if $s \step{a} t$ and $t \steps{w} u$, then $s \steps{aw} u$, and
if $s \step{\tau} t$ and $t \steps{w} u$, then $s \steps{w} u$.

We see that $\tau$-steps do not contribute to the string $w$.
We write
$s\step{}t$ for there exists $a\in\Act \cup \{\tau\}$ such that
$s\step{a}t$. Similarly, we write $s\steps{}t$ for ``there exists
$w\in\Act^{*}$ such that $s\steps{w}t$'' and say that $t$ is
\emph{reachable} from $s$. If $s\steps{w}t$ takes at least one step, we write $s\steps{w}^{+}t$.
We write $s \not\step{a}$ if there is no $t \in \mathcal{S}$ with $s \step{a} t$.
Finally, we write $s \step{(a)} t$ for ``$s \step{a} t$ or $a = \tau$ and $s = t$''.

By considering language equivalence classes of process graphs, we
recover language equivalence as a semantics, but we can also consider other
equivalence relations. Notable among these is \emph{bisimilarity}.

\begin{defi}
  Let $(\mathcal{S},\mathcal{A},\rightarrow{},{\downarrow})$ be a
  labelled transition system. A symmetric binary relation $R$ on
  $\mathcal{S}$ is a \emph{strong bisimulation} if it satisfies the following
  conditions for every $s,t\in\mathcal{S}$ such that $s\mathrel{R} t$ and for all $a\in\mathcal{A}\cup\{\tau\}$:
  \begin{enumerate}
    \item if $s\xrightarrow{a}s'$ for some $s'\in\mathcal{S}$,
      then there is a $t'\in\mathcal{S}$ such that
      $t\xrightarrow{a}t'$ and $s'\mathrel{R}t'$; and
    \item if $s{\downarrow}$, then $t{\downarrow}$.
    \end{enumerate}
    If there is a strong bisimulation relating $s$ and $t$ we write $s \bisim t$.
\end{defi}

Sometimes we can use the \emph{strong} version of bisimilarity
defined above, which does not give special treatment to
$\tau$-labelled transitions. In general, when we do give special treatment to $\tau$-labeled transitions, we use some form of \emph{branching bisimulation} \cite{GW96}.

\begin{defi}
Let $(\mathcal{S},\mathcal{A},\rightarrow{},{\downarrow})$ be a
  labelled transition system. A symmetric binary relation $R$ on
  $\mathcal{S}$ is a \emph{branching bisimulation} if it satisfies the following
  conditions for every $s,t\in\mathcal{S}$ such that $s\mathrel{R} t$ and for all $a\in\mathcal{A}\cup\{\tau\}$:
    \begin{enumerate}
    \item if $s\xrightarrow{a}s'$ for some $s'\in\mathcal{S}$,
      then there are states $t',t'' \in\mathcal{S}$ such that
      $t \steps{\epsilon} t'' \step{(a)} t'$, $s \mathrel{R} t''$ and $s'\mathrel{R}t'$; and
    \item if $s{\downarrow}$, then there is a state $t' \in \mathcal{S}$ such that $t \steps{\epsilon} t'$ and  $t'{\downarrow}$.
    \end{enumerate}
    If there is a branching bisimulation relating $s$ and $t$, we write $s \bbisim t$.
\end{defi}

In this article, we use the finest branching bisimilarity
called \emph{divergence-preserving branching
  bisimilarity}, which was introduced in \cite{GW96} (see also
\cite{GLT09} and \cite{Lut20} for an overview of recent results).

\begin{defi}
A branching bisimulation $\mathrel{R}$ is \emph{divergence-preserving} if for all $s,t \in \mathcal{S}$, whenever there is a infinite sequence of states $s_0, s_1, \ldots$ such that $s = s_0$, $s_i \step{\tau} s_{i+1}$ and $s_i \mathrel{R} t$ for all $i \geq 0$, then there is a state $t'$ with $t \steps{\epsilon}^{+}t'$ and $s_i \mathrel{R} t'$ for some $i \geq 0$.
We write $s \bbisim^{\Delta} t$ if there is a divergence-preserving branching bisimulation relating $s$ and $t$.
\end{defi}

\begin{thm}
Strong bisimilarity, branching bisimilarity and divergence-preserving branching bisimilarity are equivalence relations on labeled transition systems.
\end{thm}
\begin{proof}
See \cite{Bas96} and \cite{GLT09}.
\end{proof}

A \emph{process} is a divergence-preserving branching bisimilarity equivalence class of process graphs.

\section{Parallel Pushdown Automata}

We consider an abstract model of a computer with a memory in the form of a \emph{bag}: the bag is an unordered multiset, an element can be removed from the bag (\emph{get}), or an element can be added to it (\emph{put}).  Moreover, we can see when an element does not occur in the bag (a \emph{failed} get). This is somewhat different than the definition in \cite{Mol96}, who defined parallel pushdown automata by means of rewrite systems.

We claim our definition is a more natural one, when we compare with
the definition of a pushdown automaton. In a pushdown automaton, we
can pop the top element of the stack, or we can observe there is no
top element (i.e., the stack is empty). In  a bag, on the other hand,
all elements are directly accessible. We can pop (remove) any element,
or observe this element does not occur. Just observing that the bag is
empty, does not lead to a satisfactory theory (see \cite{vT11}).

We use notation $\mathcal{D}^{\lbag \rbag}$ for the set of bags with elements from $\mathcal{D}$. We use two disjoint copies of $\mathcal{D}$, $\mathcal{D}^{+} = \{(d,+) \mid d \in \mathcal{D}\}$ and $\mathcal{D}^{-} = \{(d,-) \mid d \in \mathcal{D}\}$. We write $+d$ instead of $(d,+)$ and $-d$ instead of $(d,-)$. We denote $\mathcal{D}^{\pm} = \mathcal{D}^{+} \cup \mathcal{D}^{-}$.

 \begin{defi}[parallel pushdown automaton]
 A \emph{parallel pushdown automaton} $M$ is a sextuple \\ $(\mathcal{S},\mathcal{A}, \mathcal{D},{\rightarrow},{\uparrow},{\downarrow})$ where:
 \begin{enumerate}
 \item $\mathcal{S}$ is a finite set of states,
 \item $\mathcal{A}$ is a finite input alphabet, $\tau \not\in \mathcal{A}$ is the unobservable step,
 \item $\mathcal{D}$ is a finite data alphabet, 
  \item $\mathalpha{\rightarrow} \subseteq \mathcal{S} \times  (\mathcal{A} \cup \{\tau\}) \times \mathcal{D}^{\pm} \times \mathcal{D}^{\lbag \rbag} \times \mathcal{S}$ is a finite set of \emph{transitions} or \emph{steps},
 \item $\mathalpha{\uparrow} \in \mathcal{S}$ is the initial state, in the initial state the bag is empty,
 \item $\mathalpha{\downarrow} \subseteq \mathcal{S}$ is the set of final or accepting states.
 \end{enumerate}
\end{defi}

 If $(s,a,+d,x,t) \in \mathalpha{\rightarrow}$ with $d \in \mathcal{D}$, we write $s \xrightarrow{a[+d/x]} t$, and this means that the machine, when it is in state $s$ and $d$ is an element of the bag, can consume input symbol $a$, replace $d$ by the bag $x$ and thereby move to state $t$. On the other hand, we write $s \xrightarrow{a[-d/x]} t$, and this means that the machine, when it is in state $s$ and the bag does not contain a $d$, can consume input symbol $a$, put $x$ in the bag and thereby move to state $t$. 
 In steps $s \xrightarrow{\tau[+d/x]} t$ and $s
 \xrightarrow{\tau[-d/x]} t$, no input symbol is consumed, only
 the bag is modified.
 
 Notice that we defined a parallel pushdown automaton in such a way that it can be detected whether or not an element occurs in the bag.

\begin{figure}[htb]
\begin{center}
\begin{tikzpicture}[->,>=stealth',node distance=2cm, node font=\footnotesize, state/.style={circle, draw, minimum size=.5cm,inner sep=0pt}]
  \node[state,initial,initial text={},initial where=left,accepting] (s0) {$\uparrow$};
  
  \path[->]
  (s0) edge[in=-45,out=45,loop]
         node[right] {$\begin{array}{c}a[-1/\lbag1\rbag]\\
                                      a[+1/\lbag1,1\rbag]\\
                                      b[+1/\emptyset]\end{array}$} (s0);
\end{tikzpicture}
\end{center}
\caption{Parallel pushdown automaton of a counter.}\label{fig:pda}
\end{figure}
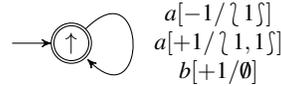

For example, consider the parallel pushdown automaton
depicted in Figure~\ref{fig:pda}. It represents the process that can
start to read an $a$, and after it has read at least one $a$, can read additional $a$'s but can also
read $b$'s. Upon acceptance, it will have read up to as many
$b$'s as it has read $a$'s.
Interpreting symbol $a$ as an increment and $b$ as a decrement, we can see this process as a \emph{counter}.
   
We do not consider the language of a parallel pushdown automaton, but rather
consider the process, i.e., the divergence-preserving branching bisimilarity equivalence class of the
process graph of a parallel pushdown automaton. A
state of this process graph is  a pair
$(s,x)$, where $s \in \mathcal{S}$ is the current state and $x \in \mathcal{D}^{\lbag \rbag}$ is the current contents of the bag. In the initial state, the bag is
empty. In a final state, acceptance can take place irrespective of the
contents of the bag. The transitions in the process graph are labeled by the inputs of the pushdown automaton or $\tau$.

\begin{defi}\label{def:pdalts}
  Let
    $M=(\mathcal{S},\mathcal{A}, \mathcal{D},{\rightarrow},{\uparrow},{\downarrow})$
  be a parallel pushdown automaton.
  The \emph{process graph}
  $\mathcal{P}(M)=(\mathcal{S}_{\mathcal{P}(M)},\mathcal{A},{\xrightarrow{}}_{\mathcal{P}(M)},{\uparrow}_{\mathcal{P}(M)},{\downarrow}_{\mathcal{P}(M)})$ associated with $M$ is
  defined as follows:
  \begin{enumerate}
  \item $\mathcal{S}_{\mathcal{P}(M)} = \{(s,x)\mid s\in\mathcal{S}\ \&\
    x\in\mathcal{D}^{\lbag \rbag}\}$;
  \item ${\xrightarrow{} _{\mathcal{P}(M)}}\subseteq
    {\mathcal{S}_{\mathcal{P}(M)}\times\mathcal{A}\cup\{\tau\}\times
      \mathcal{S}_{\mathcal{P}(M)}}$ is the least relation such that
    for all $s,s'\in\mathcal{S}$, $a\in\mathcal{A}\cup\{\tau\}$, $d\in\mathcal{D}$ and
    $x,x'\in\mathcal{D}^{\lbag \rbag}$ we have
    \begin{equation*}
      (s,{\lbag d\rbag} \cup x)\xrightarrow{a}_{\mathcal{P}(M)}(s',x' \cup x)\
        \text{if, and only if,}\ s\xrightarrow{a[+d/x']}s'\enskip;
      \end{equation*}
         \begin{equation*}
      (s,x)\xrightarrow{a}_{\mathcal{P}(M)}(s',x' \cup x)\
        \text{if, and only if, there exists $d\not\in x$ such that}\ s\xrightarrow{a[-d/x']}s'\enskip;
      \end{equation*}
  \item $\uparrow_{\mathcal{P}(M)}=(\uparrow,\emptyset)$;
 \item ${\downarrow}_{\mathcal{P}(M)}= \{(s,x)\mid s\in{\downarrow}\
    \&\ x\in\mathcal{D}^{\lbag \rbag}\}$.
   \end{enumerate}
\end{defi}

To distinguish, in the definition above, the set of states, the
transition relation, the initial state and the set of accepting states
of the parallel pushdown automaton from similar components of the associated
process graph, we have attached a subscript ${\mathcal{P}(M)}$ to
  the latter. In the remainder of this paper, we will suppress the
  subscript whenever it is already clear from the context whether a
  component of the parallel pushdown automaton or its associated process graph is meant.

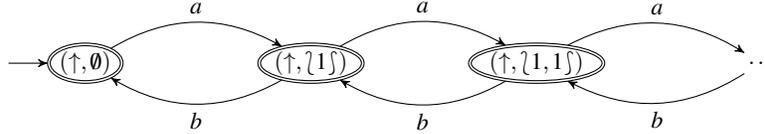
\begin{figure}[htb]
\begin{center}
\begin{tikzpicture}[->,>=stealth',node distance=3cm, node
  font=\footnotesize, state/.style={ellipse, draw, minimum size=.5cm,inner sep=0pt}]
  \node[state,accepting,initial,initial text={},initial where=left] (s0)
  {$(\uparrow,\emptyset)$};
  \node[state,accepting] [right of=s0] (s1) {$(\uparrow,\lbag1\rbag)$};
  \node[state,accepting] [right of=s1] (s2) {$(\uparrow,\lbag1,1\rbag)$};
  \node[state,accepting, draw=none] [right of=s2] (sdots) {$\dots$};

  \path[->]
    (s0) edge[bend left] node[above] {$a$} (s1)
    (s1) edge[bend left] node[above] {$a$} (s2)
    (s2) edge[bend left] node[above] {$a$} (sdots);
  \path[->]
    (sdots) edge[bend left] node[below] {$b$} (s2)
    (s2) edge[bend left] node[below] {$b$} (s1)
    (s1) edge[bend left] node[below] {$b$} (s0);
\end{tikzpicture}
\end{center}
\caption{The process graph of the counter.}\label{fig:pdatrans}
\end{figure}

Figure~\ref{fig:pdatrans} depicts the process graph associated
with the pushdown automaton depicted in Figure~\ref{fig:pda}.

In language equivalence, the definition of acceptance in parallel
pushdown automata leads to the same set of languages when we define
acceptance by final state (as we do here) and when we define
acceptance by empty bag (not considering final states). In
bisimilarity, these notions are different: acceptance by empty bag
yields a smaller set of processes than acceptance by final state. Note
that the process graph in Figure~\ref{fig:pdatrans} has infinitely
many non-bisimilar final states. It is, therefore, not bisimilar to
the process graph of a parallel pushdown automaton that accepts by
empty bag. For details, see \cite{BCLT09, vT11}.

In order to illustrate that we can realise acceptance by empty bag
also if we define acceptance by final state, consider the parallel
pushdown automaton of the counter that only accepts when empty in
Figure~\ref{counterempty}. We need three states to realise a process graph
that is divergence-preserving branching bisimilar to the process graph
in Figure~\ref{fig:pdatrans}, but with only the initial state accepting.

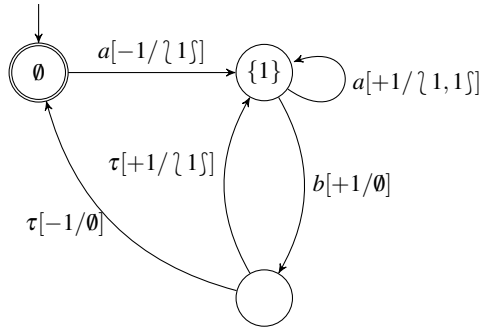
\begin{figure}[htb]
\begin{center}
\begin{tikzpicture}[->,>=stealth',node distance=3cm, node
  font=\footnotesize, state/.style={circle, draw, minimum
    size=.75cm,inner sep=0pt}]
  
  \node[state,initial,initial text={},initial where=above,accepting]
  (s0) {$\emptyset$};
  \node[state, right of=s0] (s1) {$\{1\}$};
   \node[state, below of=s1] (s4) {};
  
  \path[->]
  (s0) edge node[above] {$a[-1/\lbag 1\rbag]$} (s1)
                                   (s1) edge[bend left] node[right]
                                   {$b[+1/\emptyset]$} (s4);

   \path[->]
   (s1) edge[in=380,out=325,loop]
   node[right] {$a[+1/\lbag 1,1\rbag]$} (s1);

   \path[->]
  (s4) edge[bend left] node[left] {$\tau[-1/\emptyset]$} (s0);
 \path[->]
  (s4) edge[bend left] node[left,yshift=2ex]{$\tau[+1/\lbag 1\rbag]$} (s1);
\end{tikzpicture}
\end{center}
\caption{Counter only accepting when empty.}\label{counterempty}
\end{figure}

An important example of a parallel pushdown automaton is the bag process itself. We consider the bag that is always accepting in Figure~\ref{fig:bagalways}.
 For a given data set $\mathcal{D}$, it has actions $\mathit{ins}(d)$ (insert), $\mathit{rem}(d)$ (remove) and $\mathit{show}(^{-}d)$ (show there is no $d$). For each $d \in \mathcal{D}$, there are the transitions shown. We need the $\mathit{show}(^{-}d)$ transitions later, to indicate that no (further) remove transitions are possible. We use this in Section \ref{characterisation}.

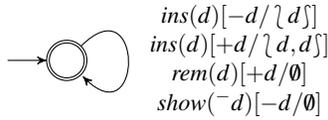
\begin{figure}[htb]
\begin{center}
\begin{tikzpicture}[->,>=stealth',node distance=2cm, node font=\footnotesize, state/.style={circle, draw, minimum size=.5cm,inner sep=0pt}]
  \node[state,initial,initial text={},initial where=left,accepting] (s0) {};
  
  \path[->]
  (s0) edge[in=-45,out=45,loop]
         node[right] {$\begin{array}{c}\mathit{ins}(d)[-d/\lbag d\rbag]\\
                                      \mathit{ins}(d)[+d/\lbag d,d\rbag]\\
                                      \mathit{rem}(d)[+d/\emptyset]\\
                                      \mathit{show}(^{-}d)[-d/\emptyset]\end{array}$} (s0);
\end{tikzpicture}
\end{center}
\caption{Parallel pushdown automaton of an always accepting bag.}\label{fig:bagalways}
\end{figure}

A parallel pushdown automaton has only finitely many transitions, so there is a maximum number of transitions from a given state, called its \emph{branching degree}. Then, also the associated process graph has a branching degree, that cannot be larger than the branching degree of the underlying parallel pushdown automaton. Thus, in a process graph associated with a parallel pushdown automaton, the branching is always \emph{bounded}. However, it is possible that its divergence-preserving branching bisimilarity equivalence class contains a process graph that is infinitely branching. Consider the following example.

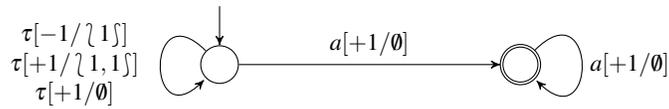
\begin{figure}[htb]
\begin{center}
\begin{tikzpicture}[->,>=stealth',node distance=4cm, node font=\footnotesize, state/.style={circle, draw, minimum size=.5cm,inner sep=0pt}]
  \node[state,initial,initial text={},initial where=above] (s0) {};
  \node[state,accepting] [right of=s0] (s1) {};
  
  \path[->]
  (s0) edge[in=225,out=135,loop]
         node[left] {$\begin{array}{c}\tau[-1/\lbag1\rbag]\\
                                      \tau[+1/\lbag1,1\rbag]\\
                                      \tau[+1/\emptyset] \end{array}$} (s0)
  (s0) edge node[above] {$a[+1/\emptyset]$} (s1)
  (s1) edge[in=-45,out=45,loop] node[right] {$a[+1/\emptyset]$} (s1);
\end{tikzpicture}
\end{center}
\caption{Parallel pushdown automaton with a divergence.}\label{fig:pdadiv}
\end{figure}

\begin{exa}
Consider the parallel pushdown automaton in Figure~\ref{fig:pdadiv}. It has a process graph consisting of two infinite rows of nodes. The nodes in the top row all have a divergence, and modulo a divergence-preserving branching bisimilarity can collaps into one node, as shown in the process graph in Figure~\ref{fig:infbranch}. This top node still needs a divergent $\tau$ loop.
\end{exa}

\begin{figure}[htb]
\begin{center}
\begin{tikzpicture}[->,>=stealth',node distance=2cm, node font=\footnotesize, state/.style={ellipse, draw, minimum size=0.5cm,inner sep=0pt}]
  \node[initial, initial text=,state, initial where=above] (a1) at (1,2) {};
  \node (a5) at (7,1) {};
  \node[state, double, double distance=1pt] (b1) at (1,0) {};
  \node[state, double, double distance=1pt] (b2) at (3,0) {};
  \node[state, double, double distance=1pt] (b3) at (5,0) {};
  \node[state, double, double distance=1pt] (b4) at (7,0) {};
  \node (b5) at (9,0) {};
  \path[->]
  (a1) edge[in=225,out=135,loop]
         node[left] {$\tau$} (a1);
  \path[->] (a1) edge [right] node {$a$} (b1);
  \path[->] (a1) edge [right]  node[xshift=0.1cm] {$a$} (b2);
  \path[->] (a1) edge [right]  node[xshift=0.1cm] {$a$} (b3);
  \path[->] (a1) edge [right]  node[xshift=0.2cm] {$a$} (b4);
  \path[->] (b4) edge [below] node {$a$} (b3);
  \path[->] (b3) edge [below] node {$a$} (b2);
  \path[->] (b2) edge [below] node {$a$} (b1);
  \path[->, dashed] (a1) edge [right] node[xshift=0.2cm] {$a$} (a5);
  \path[->, dashed] (b5) edge [below] node {$a$} (b4);
\end{tikzpicture}
\end{center}
\caption{Process graph divergence-preserving branching bisimilar to the parallel pushdown automaton with divergence.}\label{fig:infbranch}
\end{figure}
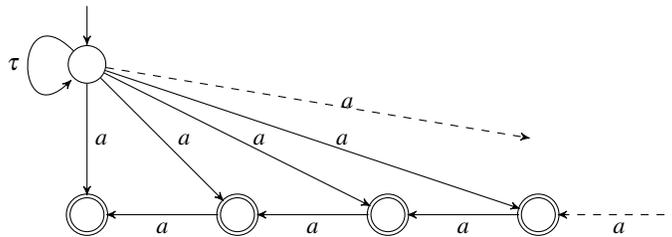

 \section{Parallel Processes}

 In the process setting, a commutative context-free grammar is a
 process algebra comprising actions, choice, parallelism and
 recursion. We start out from the process algebra \PA of \cite{BK84},
 but with sequential composition restricted to action prefixing, and
 then extended with constants $\dl$ and $\emp$ to denote deadlock and
 acceptance. We call this process algebra \BPP$^{\dl\emp}$, for Basic Parallel
 Processes with $\dl$ and $\emp$.

Let $\Act$ be a set of \emph{actions} and $\tau\not\in\Act$ \emph{the silent action}, symbols denoting atomic
events, and let $\Procids$ be a finite set of \emph{process identifiers}. The
sets $\Act$ and $\Procids$ serve as parameters of the process theory
that we shall introduce below. We use symbols $a,b,\ldots$, possibly indexed, to range over $\Act \cup \{\tau\}$, symbols $X,Y,\dots$, possibly indexed, to range over $\Procids$.
The set of \emph{parallel process expressions} is generated by the following grammar ($a\in\Act \cup \{\tau\}$, $X\in\Procids$):
\begin{equation*}
p ::= \dl \mid \emp \mid \pref[a]{p} \mid p + p \mid p \merge  p  \mid X\enskip.
\end{equation*}
The constants $\dl$ and $\emp$ respectively denote the
\emph{deadlocked} (i.e., inactive but not accepting)
process and the \emph{accepting} process. For each
$a\in\Act \cup \{\tau\}$ there is a unary action prefix operator
$\pref[a]{\_}$. We fix a finite data set $\mathcal{D}$, and actions can be parametrised with a data element.
The binary operators $+$ and $\merge$ denote
alternative composition and parallel composition,
respectively. We adopt the convention that $\pref[a]{\_}$ binds
strongest and $+$ binds weakest. 

For a (possibly empty) sequence
$p_1,\dots,p_n$ we inductively define $\sum_{i=1}^np_i=\dl$ if $n=0$
and $\sum_{i=1}^{n}p_i=(\sum_{i=1}^{n-1}p_i)+p_n$ if $n>0$.
Likewise, for a sequence
$p_1,\dots,p_n$ we inductively define $\merge_{i=1}^np_i=\emp$ if $n=0$
and $\merge_{i=1}^{n}p_i=(\merge_{i=1}^{n-1}p_i) \merge p_n$ if $n>0$.

A recursive specification over parallel process expressions is a mapping
$\Gamma$ from $\Procids$ to the set of parallel process expressions. The idea is that the process expression
$p$ associated with a process identifier $X\in\Procids$ by $\Gamma$
\emph{defines} the behaviour of $X$. We prefer to think of $\Gamma$ as a
collection of \emph{defining equations}
  $X\defeqn p$,
  exactly one for every $X\in\Procids$.
We shall, throughout the paper, presuppose a recursive specification
$\Gamma$ defining the process identifiers in $\Procids$, and we shall
usually simply write $X\defeqn p$ for $\Gamma(X)=p$. Note that, by our
assumption that $\Procids$ is finite, $\Gamma$ is finite too.

\begin{figure}[htb]
  \centering
  \begin{osrules}
        \osrule*{}{\emp \downarrow}
        \qquad \qquad
        \osrule*{}{\pref[a]p\step{a}p}
    \\
\osrule*{p \downarrow}{(p+q) \downarrow}
\quad
\osrule*{q \downarrow}{(p+q) \downarrow}
\quad
    \osrule*{p \step{a} p'}{p + q \step{a} p'}
    \quad
    \osrule*{q \step{a} q'}{p + q \step{a} q'}
  \\
 \osrule*{p \downarrow & q \downarrow}{p \merge  q \downarrow}
  \qquad
    \osrule*{p\step{a} p'}{p \merge  q \step{a} p' \merge  q}
  \qquad
    \osrule*{q \step{a} q' }{p \merge  q
      \step{a} p \merge q'}
\\
    \osrule*{p\step{a}p' & X\defeqn p}{X\step{a}p'}
  \qquad
  \osrule*{\term{p} & X\defeqn p}{\term{X}}
  \end{osrules}
\caption{Operational semantics for parallel process expressions.}
\label{fig:semantics-tspseq}
\end{figure}

We associate behaviour with process expressions by defining, on the
set of process expressions, a unary acceptance predicate $\term{}$
(written postfix) and, for every $a\in\Act \cup \{\tau\}$, a binary transition
relation $\step{a}$ (written infix), by means of the transition system
specification presented in Figure~\ref{fig:semantics-tspseq}. 

By means of these rules, the set of parallel process expressions turns into a labelled transition system, so we have strong bisimilarity, branching bisimilarity and divergence-preserving branching bisimilarity on parallel process expressions.

The operational rules presented in Fig~\ref{fig:semantics-tspseq} are
in the so-called \emph{path format} from which it immediately follows
that strong bisimilarity is a congruence \cite{BV93}. (Divergence-preserving) branching bisimilarity, however, is not a congruence, but by adding a rootedness condition we get rooted (divergence-preserving) branching bisimilarity which is a congruence \cite{GW96}. As we will not use equational reasoning in this paper, we will not use the rootedness condition.

Some recursive specifications over \BPP$^{\dl  \emp}$ will give
processes that cannot be the process of a commutative pushdown automaton.

\begin{exa} \label{exa:unguarded}
Consider the recursive equation
\[ X \defeqn a.\emp + X \merge b.\emp \enskip. \]
We show the process graph generated by the operational rules in Figure~\ref{fig:infbranch2}. As $X \step{a} \emp$, we get $X \merge b.\emp \step{a} \emp \merge b.\emp = b.\emp$ and so $X \step{a} b.\emp$. 
Continuing like this we get $X \step{a} b^n \emp$ for each $n$. Note we also have $X \step{b} X$.
\end{exa}

\begin{figure}[htb]
\begin{center}
\begin{tikzpicture}[->,>=stealth',node distance=2cm, node font=\footnotesize, state/.style={ellipse, draw, minimum size=0.5cm,inner sep=0pt}]
  \node[initial, initial text=,state, initial where=above] (a1) at (1,2) {$X$};
  \node (a5) at (7,1) {};
  \node[state, double, double distance=1pt] (b1) at (1,0) {$\emp$};
  \node[state] (b2) at (3,0) {$b.\emp$};
  \node[state] (b3) at (5,0) {$b.b.\emp$};
  \node[state] (b4) at (7,0) {$b.b.b.\emp$};
  \node (b5) at (9,0) {};
  \path[->]
  (a1) edge[in=225,out=135,loop]
         node[left] {$b$} (a1);
  \path[->] (a1) edge [right] node {$a$} (b1);
  \path[->] (a1) edge [right]  node[xshift=0.1cm] {$a$} (b2);
  \path[->] (a1) edge [right]  node[xshift=0.1cm] {$a$} (b3);
  \path[->] (a1) edge [right]  node[xshift=0.2cm] {$a$} (b4);
  \path[->] (b4) edge [below] node {$b$} (b3);
  \path[->] (b3) edge [below] node {$b$} (b2);
  \path[->] (b2) edge [below] node {$b$} (b1);
  \path[->, dashed] (a1) edge [right] node[xshift=0.2cm] {$a$} (a5);
  \path[->, dashed] (b5) edge [below] node {$b$} (b4);
\end{tikzpicture}
\end{center}
\caption{Process graph of the recursive specification of Example~\ref{exa:unguarded}.}\label{fig:infbranch2}
\end{figure}
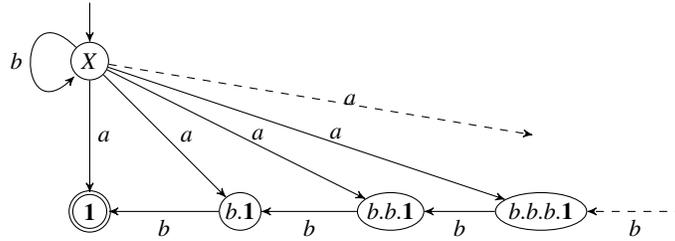

\begin{thm} \label{thm:infbranch}
The process graph of Figure~\ref{fig:infbranch2} is not divergence-preserving branching bisimilar to the process graph of any parallel pushdown automaton.
\end{thm}

To exclude recursive specifications over \BPP$^{\dl\emp}$ that give
rise to process graphs with states that necessarily have infinitely many outgoing
transitions, it suffices to formulate a standard \emph{guardedness}
condition for recursive specifications.

\begin{defi}
We say a recursive specification is \emph{weakly guarded} if every occurrence of a
process identifier in the definition of some (possibly different)
process identifier occurs within the scope of an action prefix from $\mathcal{A} \cup \{\tau\}$, and \emph{strongly guarded} if every occurrence of a process identifier in the definition of some process identifier occurs within the scope of an action prefix from $\mathcal{A}$.
\end{defi}

We will show that every finite weakly guarded recursive specification over \BPP$^{\dl\emp}$ yields a parallel pushdown automaton. We first consider a couple of examples.

\begin{exa} \label{exa:difference}
  Consider the recursive specification
  \begin{equation*}
    AC \defeqn \emp + \pref[a]{(AC \merge (\emp + b.\emp))} \enskip.
  \end{equation*}
  By following the operational rules, we obtain a process graph that
  is bisimilar to the one shown in Figure~\ref{fig:pdatrans}, and thus we obtain the parallel pushdown automaton in Figure~\ref{fig:pda}. This is the always accepting counter.

If, instead, we use the equation 
 \begin{equation*}
    EC \defeqn \emp + \pref[a]{(EC \merge b.\emp)} \enskip.
  \end{equation*}
we get the counter that only accepts when it is empty, see the parallel pushdown automaton in Figure~\ref{counterempty}.
Now we can generalize the equation of $AC$ to the following
 \begin{equation*}
    \mathit{AB} \defeqn \emp + \sum_{d \in \mathcal{D}} \pref[ins(d)]{(\mathit{AB} \merge (\emp + rem(d).\emp))} \enskip.
  \end{equation*}
We see that this is a specification of the bag. However, this bag does not have the $\mathit{show}(^{-}d)$ actions to signal that a $d$ does not occur in the bag. In fact, we will show that there is no finite weakly guarded specification over \BPP$^{\dl\emp}$ that gives rise to the parallel pushdown automaton in Figure \ref{fig:bagalways}. For now, we first look at the other direction, to show that a finite weakly guarded specification over \BPP$^{\dl\emp}$ yields the process of a parallel pushdown automaton.
\end{exa}

Since we have weakly guarded recursion, we can bring every \BPP$^{\dl\emp}$-term into head normal form. The following result uses strong bisimulation, not branching bisimulation.

\begin{thm} \label{thm:commhnf}
Let $\Gamma$ be a weakly guarded \BPP$^{\dl\emp}$-specification. Every process expression $p$ can be brought into \emph{head normal form}, i.e. there are $a_i \in \mathcal{A} \cup \{\tau\}$ and process expressions $p_i$ such that
\begin{equation*}
p \bisim (\emp +) \sum_{i=1}^n a_i.p_i
\end{equation*}
where the $\emp$ summand may or may not occur.
\end{thm}

As a result, we can bring every guarded recursive specification into Greibach Normal Form.

\begin{defi}
A recursive specification $\Gamma$ is in \emph{Greibach Normal Form} if every equation has the form $X \defeqn (\emp +) \sum_{i=1}^n a_i.\xi_i$ for actions $a_i \in \mathcal{A} \cup \{\tau\}$, where each $\xi_i$ is a parallel composition of identifiers of $\Gamma$, and $n \geq 0$.
\end{defi}

\begin{thm}
Let $\Gamma$ be a weakly guarded \BPP$^{\dl\emp}$-specification over identifiers $\mathcal{P}$. Then there is a finite set of identifiers $\mathcal{Q}$ with $\mathcal{P} \subseteq \mathcal{Q}$ and a recursive specification in Greibach Normal Form $\Delta$ over identifiers $\mathcal{Q}$ such that for all $X,Y \in \mathcal{P}$ we have $X \bisim Y$ with respect to $\Gamma$ if, and only if, $X \bisim Y$ with respect to $\Delta$. 
\end{thm}

Now we are ready to prove the main result of this section.

\begin{thm} \label{bpptoppda}
Every weakly guarded recursive specification over \BPP$^{\dl\emp}$ has a process graph that is divergence-preserving branching bisimilar to the process graph of a parallel pushdown automaton.
\end{thm}
\begin{proof}
Without loss of generality, we can assume the specification is in Greibach Normal Form. Then, all states in the generated process graph are given by a parallel composition of identifiers of the specification.
Divide the identifiers of the specification into the accepting identifiers $\mathbb{A}$ (that have a $\emp$ summand) and the non-accepting identifiers $\mathbb{N}$ that do not have a $\emp$ summand. A state in the generated process graph is accepting iff all identifiers in the parallel composition are from $\mathbb{A}$. In the parallel pushdown automaton to be constructed, we need to keep track when the last element of $\mathbb{N}$ is removed, in order to switch to an accepting state.

We take the data set $\mathcal{D}$ to be the set of identifiers of the specification. $S$ is the initial identifier. In the states of the parallel pushdown automaton, we will encode whether or not there is an element of $\mathbb{N}$, so there is a state for each subset (not multiset) of $\mathbb{N}$.
As inspiration, we use the parallel pushdown automaton of the counter that only accepts when empty in Figure~\ref{counterempty}.

%
The states of the parallel pushdown automaton are as follows:
\begin{itemize}
\item $N$, for $N \subseteq \mathbb{N}$ (a subset, not a submultiset). The bisimulation will relate state $(N,x \cup y)$ for any multiset $x \in \mathbb{A}^{\lbag \rbag}$ to the parallel composition of the elements of $x \cup y$, if $y$ contains all elements of $N$ and no other non-accepting identifiers.
\item There is an auxiliary state $N_X$, for each $N \subseteq \mathbb{N}$ and $X \in N$.
\end{itemize}
The initial state is $\emptyset$. $\emptyset$ is the only accepting state.

Now the steps:
\begin{enumerate}
\item $\emptyset \xrightarrow{a[-S/\xi]} \emptyset$, whenever $S$ has a summand $a.\xi$ and $\xi$ has only accepting identifiers.
\item $\emptyset \xrightarrow{a[-S/\xi]} N$ whenever $S$ has a summand $a.\xi$ and $\xi$ has at least one non-accepting identifier. $N$ collects the non-accepting identifiers of $\xi$.
\item $N \xrightarrow{a[+X/\xi]} N'$, whenever $X \not\in N$ is accepting, $X$ has a summand $a.\xi$ and $N'$ unites $N$ with the non-accepting identifiers of $\xi$.
\item $N \xrightarrow{a[+X/\xi]} N'$, whenever $X \in N$ has a summand $a.\xi$, $\xi$ has at least one non-accepting identifier and $N'$ unites $N$ with the non-accepting identifiers of $\xi$.
\item if $X \in N \subseteq \mathbb{N}$, $X$ has a summand $a.\xi$ with all identifiers in $\xi$ accepting, add three transitions
\[ N \xrightarrow{a[+X/\xi]} N_X \xrightarrow{\tau[+X/\lbag X \rbag]} N .\]
and
\[ N_X \xrightarrow{\tau[-X/\emptyset]} N-\{X\}. \]
\end{enumerate}
Notice that all the added $\tau$-steps in the transition system are inert, as from the added $N_X$ states exactly one transition can be taken, depending on whether or not $X$ occurs in the parallel composition.
\end{proof}

\begin{figure}[htb]
\begin{center}
\begin{tikzpicture}[->,>=stealth',node distance=3cm, node
  font=\footnotesize, state/.style={circle, draw, minimum
    size=.75cm,inner sep=0pt}]
  
  \node[state,initial,initial text={},initial where=above,accepting]
  (s0) {$\emptyset$};
  \node[state, right of=s0] (s1) {$\{B\}$};
  \node[state, below of=s1] (s3) {$$};
  
  \path[->]
  (s0) edge[in=220,out=145,loop]
         node[left] {$\begin{array}{l}c[-S/\lbag S, D\rbag]\\
                                      c[+S/\lbag S, D\rbag]\\
                                      d[+D/\emptyset]\end{array}$}
                                  (s0)
  (s0) edge node[above] {$\begin{array}{c}
                                                            a[+S/\lbag
                                       S,B\rbag]\\
                                                           a[-S/\lbag
                                       S,B\rbag]\end{array}$} (s1);     

   \path[->]
   (s1) edge[in=380,out=325,loop]
   node[right] {$\begin{array}{l}a[+S/\lbag S,B\rbag]\\
                                      c[+S/\lbag S, D\rbag]\\
                                      d[+D/\emptyset]\end{array}$} (s1)
   (s1) edge[bend left] node[right] {$b[+B/\emptyset]$} (s3);

   \path[->]
   (s3) edge[bend left] node[left] {$\tau[-B/\emptyset]$} (s0);

   \path[->]
   (s3) edge[bend left] node[left,yshift=2ex]{$\tau[+B/\lbag B\rbag]$} (s1);
\end{tikzpicture}
\end{center}
\caption{Parallel pushdown automaton of Example \ref{exa:largeppda}.}\label{fig:largeppda}
\end{figure}
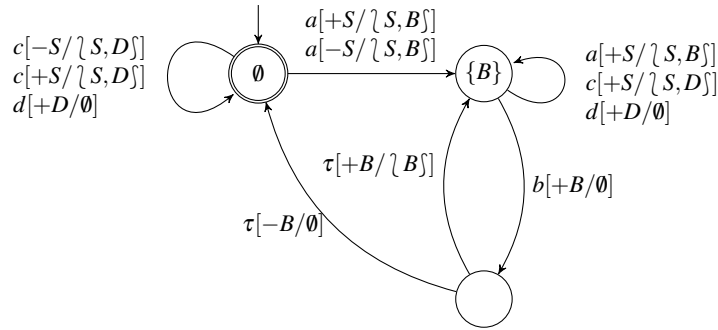

\begin{exa}\label{exa:largeppda}
Let the guarded recursive specification $\Gamma$ be given as follows. Notice it is in Greibach Normal Form, and $\mathbb{N} = \{B\}, \mathbb{A} = \{S,D\}$.
\[ S \defeqn \emp + a.(S \merge B) + c.(S \merge D)	\qquad B \defeqn b.\emp \qquad D \defeqn \emp + d.\emp \]
Following the proof of Theorem~\ref{bpptoppda} results in the parallel pushdown automaton shown in Figure \ref{fig:largeppda}.
Notice the similarity with the parallel pushdown automaton shown in Figure~\ref{counterempty}.
\end{exa}

As we stated, the other direction does not work: we cannot find a
finite weakly guarded \BPP$^{\dl \emp}$-specification for the
one-state parallel pushdown automaton of the always accepting bag in
Figure \ref{fig:bagalways}. It is technically somewhat simpler to
prove such a negative result for the
one-state parallel pushdown automaton shown in Figure~\ref{fig:ppda.ac}.

\begin{figure}[htb]
\begin{center}
\begin{tikzpicture}[->,>=stealth',node distance=2cm, node font=\footnotesize, state/.style={circle, draw, minimum size=.5cm,inner sep=0pt}]
  \node[state,initial,initial text={},initial where=left,accepting] (s0) {$\uparrow$};
  
  \path[->]
  (s0) edge[in=-45,out=45,loop]
         node[right] {$\begin{array}{c}c[-1/\lbag1\rbag]\\
                                      a[+1/\lbag1,1\rbag]\\
                                      b[+1/\emptyset]\end{array}$} (s0);
\end{tikzpicture}
\end{center}
\caption{Parallel pushdown automaton of a counter with a change.}\label{fig:ppda.ac}
\end{figure}
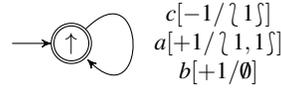

\begin{thm}\label{thm:one-statenospec}
For the one-state parallel pushdown automaton in Figure~\ref{fig:ppda.ac} there is no finite weakly guarded \BPP$^{\dl\emp}$ specification such that their process graphs are divergence-preserving branching bisimilar.
\end{thm}

Now let us reconsider the parallel pushdown automaton of the bag. The problem is, that a $show(^{-}d)$-action can only occur if no $rem(d)$-action can occur.
Thus, in a sum context, the $show(^{-}d)$ action should have
\emph{lower priority} than the $rem(d)$ action. In general, we assume
we have a partial ordering $<$ on $\mathcal{A} \cup \{\tau\}$, where
$a < b$ means that $a$ has lower priority than $b$, satisfying that
$\tau < a$ never holds, and whenever $a < b$ then also $a < \tau$. The
\emph{priority operator} $\theta$ will implement the priorities, and
is given by the operational rules in Figure~\ref{fig:prior}, see
\cite{BBK86,BBR10}. Notice that the second rule for the priority
operator uses a negative premise. Transition system specifications
with negative premises may, in general, not define a unique transition
relation that agrees with provability from the transition system
specification, but our restriction to weakly guarded specifications
eliminates this problem \cite{Gro93, BG96,Gla04}. Also, note that
(rooted) branching bisimilarity is not compatible with the priority operator, but
divergence-preserving branching bisimilarity is \cite{FGL19}.

\begin{figure}[htb]
  \centering
  \begin{osrules}
     \osrule*{p \downarrow}{\theta(p) \downarrow}
  \qquad
    \osrule*{p\step{a} p' & \forall b>a \quad p\not\step{b}}{\theta(p) \step{a} \theta(p')} 
    \\
      \osrule*{p\step{a}p'}{\rho_f(p) \step{f(a)} \rho_f(p')}
  \qquad
  \osrule*{\term{p}}{\term{\rho_f(p)}}
  \end{osrules}
\caption{Operational semantics for priorities and renaming.}
\label{fig:prior}
\end{figure}

With the help of this operator, we can give the following specification of the always accepting bag, assuming $show(^{-}d) <rem(d)$.
\begin{equation*}
    \mathit{ABag} \defeqn \emp + \sum_{d \in \mathcal{D}} \pref[ins(d)]{\theta(\mathit{ABag} \merge (\emp + rem(d).\emp)) + \sum_{d \in \mathcal{D}} show(^{-}d).\mathit{ABag}} \enskip.
  \end{equation*}
For the parallel pushdown automaton in Figure~\ref{fig:ppda.ac}, it is enough to take $c < b$.
In general, we need to put a priority ordering on the labels of a parallel pushdown automaton. This may not be possible if some labels are the same, or if a $\tau$ occurs as a label. Therefore, we need to ensure all the labels in the parallel pushdown automaton are distinct and from $\mathcal{A}$, in order to be able to impose a priority ordering.

Thus, given a parallel pushdown automaton, we consider another parallel pushdown automaton with distinct labels, solve the problem for that automaton, and then rename the labels again to their original values. This renaming is done by a \emph{renaming operator} $\rho_f$, where $f$ is any function on $\mathcal{A} \cup \{\tau\}$ satisfying $f(\tau) = \tau$. The renaming operator has the operational rules shown in Figure~\ref{fig:prior}, see \cite{BB88,BBR10}.

Now we extend \BPP$^{\dl\emp}$ to include the priority operator and
renaming operators. We call this extended algebra
\BPP$^{\dl\emp}_{\theta}$. Theorem~\ref{bpptoppda} can be extended to
\BPP$^{\dl\emp}_{\theta}$.

  \begin{thm} \label{bpptoppdatheta}
Every weakly guarded recursive specification over \BPP$^{\dl\emp}_{\theta}$ has a process graph that is divergence-preserving branching bisimilar to the process graph of a parallel pushdown automaton.
\end{thm}

In the other direction, it works for every one-state parallel pushdown automaton.

\begin{thm}\label{onestatethm}
For every one-state parallel pushdown automaton  there is a finite weakly guarded \BPP$^{\dl\emp}_{\theta}$ specification such that their process graphs are divergence-preserving branching bisimilar.
\end{thm}

Thus, for all one-state parallel pushdown automata we can find a specification in \BPP$^{\dl\emp}_{\theta}$. This result does not extend to parallel pushdown automata with more than one state.

\begin{figure}[htb]
\begin{center}
\begin{tikzpicture}[->,>=stealth',node distance=4cm, node font=\footnotesize, state/.style={circle, draw, minimum size=.5cm,inner sep=0pt}]
  \node[state,initial,initial text={},initial where=above] (s0) {};
  \node[state,accepting] [right of=s0] (s1) {};
  
  \path[->]
  (s0) edge[in=225,out=135,loop]
         node[left] {$\begin{array}{c}a[-1/\lbag1\rbag]\\
                                      a[+1/\lbag1,1\rbag]\\
                                      b[+1/\emptyset] \end{array}$} (s0)
  (s0) edge node[above] {$c[+1/\emptyset]$} (s1)
  (s1) edge[in=-45,out=45,loop] node[right] {$d[+1/\emptyset]$} (s1);
\end{tikzpicture}
\end{center}
\caption{Parallel pushdown automaton that cannot be specified in \BPP$^{\dl\emp}_{\theta}$.}\label{fig:tegenvb}
\end{figure}

\begin{thm} \label{nospec}
There is a parallel pushdown automaton with two states, such that there is no weakly
guarded \BPP$^{\dl\emp}_{\theta}$ specification with the same process.
\end{thm}

In order to recover the correspondence between parallel pushdown automata and parallel process algebra, we need to add communication with value passing. 

\section{Communicating processes}
We extend the basic parallel processes \BPP$^{\dl\emp}$ by adding a communication mechanism. We assume we have a finite set of communication ports $\mathcal{C}$, and that each parametrised action $c(d)$ is the result of the communication of the send action $c!d$ and the receive action $c?d$. The data set $\mathcal{D}$ is finite. 
Define $\mathit{COM}_C =  \{c!d, c?d \mid c \in C, d \in \mathcal{D}\}$ for a set of ports $C \subseteq \mathcal{C}$. The \emph{encapsulation operator} $\encap{C}{}$ will block the send and receive actions from the set of ports $C$. The \emph{abstraction operator} $\tau_C$ will hide all parametrised actions from the set of ports $C$.

\begin{figure}[htb]
  \centering
  \begin{osrules}
\osrule*{p \step{c!d} p' & q \step{c?d} q'}{p \merge  q \step{c(d)} p' \merge q' \quad q \merge p \step{c(d)} q' \merge p'}
 \\
\osrule*{p \downarrow}{\encap{C}{p} \downarrow} \qquad \osrule*{p \step{a} p' & a \not\in \mathit{COM}_C}{\encap{C}{p} \step{a} \encap{C}{p'}} \\
 \osrule*{p \downarrow}{\tau_{C}(p) \downarrow} \qquad \osrule*{p \step{c(d)} p' & c \in C}{\tau_{C}(p) \step{\tau} \tau_C(p')} \qquad \osrule*{p \step{a} p' & a \neq c(d) \mbox{ for } c \in C}{\tau_C(p) \step{a} \tau_C(p')}
   \end{osrules}
\caption{Operational semantics for communication, encapsulation and abstraction.}
\label{fig:semantics-seqc}
\end{figure}

The process algebra \BCP$^{\dl\emp}$ extends \BPP$^{\dl\emp}$ with communication, encapsulation and abstraction; likewise, \BCP$^{\dl\emp}_{\theta}$ extends \BPP$^{\dl\emp}_{\theta}$.
Using communication, we can specify communicating bags: $\mathit{ABag}^{io}$ defines the always accepting bag with input port $i$ and output port $o$, while $\mathit{EBag}^{io}$ defines the bag with input port $i$ and output port $o$ that is only accepting when it is empty. We see both the $rem(d)$ actions and the $show(^{-}d)$ actions as outputs.

 \[   \mathit{ABag}^{io} \defeqn \emp + \sum_{d \in \mathcal{D}} \pref[i?d]{\theta( \mathit{ABag}^{io} \merge (\emp + o!(^{+}d).\emp))} + \sum_{d \in \mathcal{D}} o!(^{-}d).\mathit{ABag}^{io} \]
 \[    \mathit{EBag}^{io} \defeqn \emp + \sum_{d \in \mathcal{D}} i?d.\theta( \mathit{EBag}^{io} \merge o!(^{+}d).\emp) + \sum_{d \in \mathcal{D}} o!(^{-}d).\mathit{EBag}^{io}\enskip. \]
 
 In Section~\ref{characterisation}, we will use the communicating always accepting bag to make the communication between a finite control and a memory in the form of a bag explicit. Here, we restate a classical result: putting two bags with unrestricted capacity in series will again be a bag with unrestricted capacity.
 \[ \mathit{ABag}^{io} \bbisim^{\Delta} \tau_{\{\ell\}}(\encap{\{\ell\}}{\mathit{ABag}^{i\ell} \merge \mathit{ABag}^{\ell o}})  \qquad
\mathit{EBag}^{io} \bbisim^{\Delta} \tau_{\{\ell\}}(\encap{\{\ell\}}{\mathit{EBag}^{i\ell} \merge \mathit{EBag}^{\ell o}}) \]
 
 Again, we can bring every \BCP$^{\dl\emp}_{\theta}$-term into head normal form.

\begin{thm}
Let $\Gamma$ be a weakly guarded \BCP$^{\dl\emp}_{\theta}$-specification. Every process expression $p$ can be brought into \emph{head normal form}, i.e. there are $a_i \in \mathcal{A} \cup \{\tau\}$ and process expressions $p_i$ such that
\begin{equation*}
p \bisim (\emp +) \sum_{i=1}^n a_i.p_i
\end{equation*}
where the $\emp$ summand may or may not occur.
\end{thm}
\begin{proof}
By induction on the structure of $p$ (see \cite{BBR10}). We use Milner's Expansion Law, now with communication.
\end{proof}

As a result, we can bring every guarded recursive specification into Greibach Normal Form.

\begin{defi}
A recursive specification $\Gamma$ over \BCP$^{\dl\emp}_{\theta}$ is in \emph{Greibach Normal Form} if every equation has the form 
\begin{equation*}
X \defeqn (\emp +) \sum_{i=1}^n a_i.\tau_C(\encap{C}{\theta(\xi_i)}).
\end{equation*}
for actions $a_i \in \mathcal{A} \cup \{\tau\}$, where each $\xi_i$ is a parallel composition of identifiers of $\Gamma$, and $n \geq 0$.
\end{defi}

\begin{thm}
Let $\Gamma$ be a weakly guarded \BCP$^{\dl\emp}_{\theta}$-specification over identifiers $\mathcal{P}$. Then there is a finite set of identifiers $\mathcal{Q}$ with $\mathcal{P} \subseteq \mathcal{Q}$ and a recursive specification in Greibach Normal Form $\Delta$ over identifiers $\mathcal{Q}$ such that for all $X,Y \in \mathcal{P}$ we have $X \bisim Y$ with respect to $\Gamma$ if, and only if, $X \bisim Y$ with respect to $\Delta$. 
\end{thm}

\section{The full correspondence}

With the help of value-passing communication, we can now establish our main result: for every parallel pushdown automaton we can find a specification in \BCP$^{\dl\emp}_{\theta}$. The communication actions will pass on the information of the current state of the parallel pushdown automaton. Let us look at the parallel pushdown automaton in Figure~\ref{fig:tegenvb}, that did not have a finite specification in \BCP$^{\dl\emp}_{\theta}$.

\begin{exa}
Consider the parallel pushdown automaton in Figure~\ref{fig:tegenvb}, with initial state $s$ and accepting state $t$. We need to distinguish between the two $a$-actions on state $s$, let us call them $a^{-}$ and $a^{+}$.
Action $a^{-}$ has lower priority than $a^{+}$.
We just need to communicate in which of the states we are, so we use actions $p!s, p!t, p?s$ and $p?t$ for some communication port $p$. Actions $p(s),p(t)$ have the highest priority.
As all components in a parallel composition need the state information, we need to communicate the state information repeatedly, until all components are brought into the right position. 
After this, we need to exit the communication process. Define $P_s \defeqn \emp + p!s.P_s + exit.\emp$ and $P_t \defeqn \emp + p!t.P_t + exit.\emp$ and $p(s) > exit, p(t) > exit$ and $exit > e$ for $e \in \{a^{+},a^{-},b,c,d \}$.
\[ S \defeqn a^{-}.\tau_p(\encap{p}{\theta(P_s \merge X_0 \merge X_1)}) \]
\[ X_1 \defeqn p?s.(a^{+}.(P_s \merge X_1 \merge X_1) + b.P_s + c.P_t) + p?t.(\emp + d.P_t)  \]
\[ X_{0} \defeqn p?s.a^{-}.(P_s \merge X_1) + p?t.\emp \]
\end{exa}

\begin{thm}\label{thm:fullcorr}
For every parallel pushdown automaton there is a finite weakly guarded specification over \BCP$^{\dl\emp}_{\theta}$ such that their process graphs are divergence-preserving branching bisimilar.
\end{thm}

\begin{thm}
For every finite weakly guarded \BCP$^{\dl\emp}_{\theta}$-specification there is a parallel pushdown automaton such that their process graphs are divergence-preserving  branching bisimilar.
\end{thm}
\begin{proof}
As again, we can bring a finite weakly guarded \BCP$^{\dl\emp}_{\theta}$-specification into Greibach Normal Form, this proof goes along the lines of the proof of Theorem~\ref{bpptoppda}. The only difference is, is that because of a communication action, \emph{two} non-accepting identifiers can be removed from a parallel composition at the same time.
\end{proof}

To conclude this section, we consider how far we can go with communication, but without priorities. In the bag, not using priorities, it is required to count the number of remove transitions, in order to know when a show absence transition is enabled. We can do this counting in the communication actions, but then the parametrising data set $\mathcal{D}$ becomes infinite, and the specification uses countable sums. This is a drawback, in our opinion.

\begin{thm}\label{thm:ppdabcp}
For every parallel pushdown automaton there is a finite weakly guarded specification over \BCP$^{\dl\emp}$ extended with infinite choice, such that their process graphs are divergence-preserving branching bisimilar.
\end{thm}

\begin{exa}
Consider the parallel pushdown automaton in Figure~\ref{fig:ppda.ac}, with state $s$. As the data set is a singleton, we just need to count the number of $1$'s, and we use natural numbers as parameters.
\[ S \defeqn \emp + c.\tau_s(\encap{s}{s!1.\emp \merge X_{\{1\}}}) \]
\[ X_{\{1\}} \defeqn s?1.(\emp + (a.s!2.\emp \merge X_{\{1\}} \merge X_{\{1\}}) + (b.s!0.\emp \merge X_{\emptyset})) + \]
\[ \qquad \qquad + \sum_{n \geq 2} s?n.(\emp + (a.s!(n+1).\emp \merge X_{\{1\}} \merge X_{\{1\}}) + b.s!(n-1).\emp) \]
\[ X_{\emptyset} \defeqn s?0.(\emp + (c.s!1.\emp \merge X_{\{1\}})). \]
\end{exa}

\section{A characterisation}\label{characterisation}

A computer shows interaction between a finite control and the memory. The finite control can be represented by a regular process (a finite automaton). 
In \cite{BCT08}, we considered a memory in the form of a stack, and we established that a pushdown process can be characterised as a regular process communicating with a stack. Here, we have a memory in the form of a bag, and we can establish a similar result.

\begin{thm}\label{thm:charbag}
A process $p$ is a parallel push-down process, if and only there is a
regular process $q$ such that
  $p \bbisim^{\Delta} \tau_{\{i,o\}}(\encap{\{i,o\}}{q \merge
    \mathit{ABag}^{io}})$.
\end{thm}

  In \cite{BCT09}, it was established that every parallel process
  expression is rooted branching bisimilar to a regular process
  communicating with a process that is defined as follows:
   \[   \mathit{AB}^{io} \defeqn \emp + \sum_{d \in \mathcal{D}}
     \pref[i?d]{ \mathit{AB}^{io} \merge (\emp + o!d.\emp)}. \]
  By Theorem~\ref{bpptoppda}, every parallel process expression
  denotes a parallel pushdown process, and so
  Theorem~\ref{thm:charbag} can be applied to get a characterisation
  in terms of a regular process that communicates with
  $\mathit{ABag}^{io}$, the aways accepting bag. Note, however, that
  $\mathit{ABag}^{io}$ uses the priority operator to facilitate the show
  absence actions. With the process $\mathit{AB}^{io}$ from
  \cite{BCT09} we can establish a similar result, but we have to replace receiving an element from a bag by \emph{getting} an element from a bag, where failure to get a particular element can be detected (see \cite{Ber85,BBR10}. Thus, for $d \in \mathcal{D}$, we add elements $c?\!?^{+}d$ (a get) and $c?\!?^{-}d$ (a failed get) to $\mathit{COM}_C$ with $c \in C$, and add $c \times d$ for a failed communication that will also be hidden by $\tau_C$ with $c \in C$. We add the operational rules in Figure~\ref{fig:semantics-getc}. Notice the  second rule uses a negative premise. Still, as we use weakly guarded recursion, and the rules are in so-called \emph{panth} format, we obtain a labelled transition system, see \cite{Ver95}.

\begin{figure}[htb]
  \centering
  \begin{osrules}
   \osrule*{ p \step{c?\!?^{+}d} p' & q \step{c!d} q'}{p \merge  q \step{c(d)} p' \merge q' \quad q \merge p \step{c(d)} q' \merge p'} \qquad
    \osrule*{p \step{c?\!?^{-}d} p' &  q \nstep{c!d}}{p \merge q \step{c\times d} p' \merge q \quad q \merge p \step{c\times d} q \merge p'}  \\
   \end{osrules}
\caption{Operational semantics for get communication.}
\label{fig:semantics-getc}
\end{figure}

\begin{thm}\label{thm:charget}
A process $p$ is a parallel push-down process, if and only there is a regular process $q$ such that
$p \bbisim^{\Delta} \tau_{\{i,o\}}(\encap{\{i,o\}}{q \merge
    \mathit{AB}^{io}})$.
\end{thm}

We see the bag is the prototypical parallel pushdown process, as all
parallel pushdown processes can be realised as a regular process
communicating with a bag. A bag is not a pushdown process. Likewise,
the stack is the prototypical pushdown process, but not a parallel
pushdown process. The counter is not a regular process, but it is both
a pushdown process and a parallel pushdown
process. Figure~\ref{fig:classification} provides a complete
picture. The queue is not a pushdown process and also not a parallel
pushdown process. It is the prototypical executable process, as every
executable process can be characterized as a regular process
communicating with an always accepting queue. By using a queue, the
Turing tape can be defined.

\begin{figure}[htb]
\begin{center}
\begin{tikzpicture}
\usetikzlibrary{shapes}
\clip (0,0) circle (3.3cm);
\node[circle,minimum size=6.6cm,draw] (exe) at (0,0){};
\path (exe.north) coordinate (exetext) node[below,yshift=-0.2cm]{Executable};
\node[circle,minimum size=2.2cm,draw, anchor=south] (reg) at (exe.south){Regular};
\node[ellipse, minimum width=5.61cm, minimum height=3.52cm,draw,anchor=west,rotate=45] (ppd) at (reg.south west){};
\node[circle,minimum size=1.1cm,anchor=east,yshift=.3cm] (ppdtext) at
(exe.east){Parallel Pushdown};
\node[ellipse, minimum width=5.61cm, minimum height=3.52cm,draw,anchor=east,rotate=-45] (pd) at (reg.south east){};
\node[circle,minimum size=1.1cm,anchor=west,yshift=.3cm] (pdtext) at (exe.west){Pushdown};
\path (0,-0.5) coordinate (C) node[below]{Counter};
\fill (C) circle(2pt);
\path (2,-0.5) coordinate (B) node[below]{Bag};
\fill (B) circle(2pt);
\path (-2,-0.5) coordinate (S) node[below]{Stack};
\fill (S) circle(2pt);
\path (0,1.5) coordinate (Q) node[below]{Queue};
\fill (Q) circle(2pt);
\end{tikzpicture}
\end{center}
  \caption{Classification of Executable, Pushdown, Parallel Pushdown,
    and Regular processes and the prototypical processes Queue, Bag,
    Stack and Counter.}\label{fig:classification}
\end{figure}
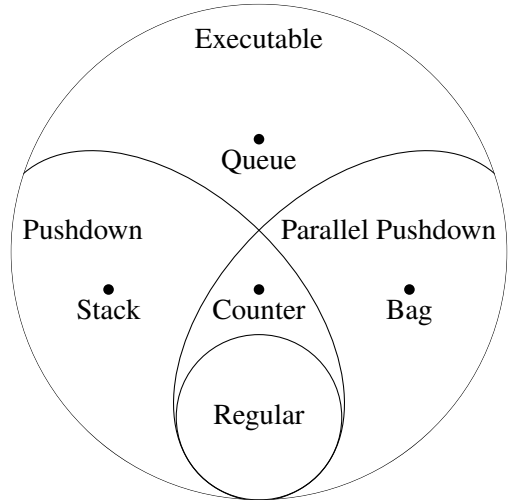

\section{Conclusion}
In language theory, the set of languages given by
a parallel pushdown automaton coincides with the set of languages given by a
commutative context-free grammar. A language is an equivalence class of
process graphs modulo language
equivalence. A process is an equivalence class of process graphs modulo divergence-preserving branching bisimulation.

This paper solves the question how we can characterize the set of processes given by a parallel pushdown automaton. In the process setting, a commutative context-free grammar is a process algebra with actions, choice, parallel composition and finite recursion. We need to limit to weakly guarded recursion in the process setting. Starting out from the seminal process algebra PA of \cite{BK84} with sequential composition restricted to action prefixing, we need to add constants for the inactive and accepting process and for the inactive non-accepting (deadlock) process. Thus, we arrive at the process algebra \BPP$^{\dl\emp}$ of the basic parallel processes. We extend this algebra with the priority operator $\theta$, in order to give some actions priority over others.

Then, every finite weakly guarded \BPP$^{\dl\emp}_{\theta}$ specification yields the process of a parallel pushdown automaton, but not the other way around, there are processes of parallel pushdown automata that cannot be given by a finite weakly guarded \BPP$^{\dl\emp}_{\theta}$ specification. For parallel pushdown automata with just one state, such a specification can be found.

We obtain a complete correspondence by adding value passing communication. 

\medskip

\fbox{\parbox{.9\textwidth}{The set of processes given by a parallel pushdown automaton coincides with the set of processes given by a finite weakly guarded recursive specification over a process algebra with actions, choice, priorities, and parallel composition with value passing communication.}}

\medskip

We also provide another characterisation of parallel pushdown processes: a process is a parallel pushdown process if and only if there is a regular process such that the process is divergence-preserving branching bisimilar to the regular process communicating with an always accepting bag. 

This paper contributes to our ongoing project to integrate automata theory and process theory. As a result, we can present the foundations of computer science using a computer model with interaction. Such a computer model relates more closely to the computers we see all around us.

As future work, we need to compare the algebra used here with Petri nets, see e.g. \cite{DMG12}.



\bibliographystyle{eptcs}
\bibliography{main}

\end{document}